\documentclass{article}
\usepackage{cite}
\usepackage{amsmath,amssymb,amsfonts, amsthm}
\usepackage{algorithmic}
\usepackage{graphicx}
\usepackage{textcomp}
\usepackage{xcolor}
\usepackage{hyperref}
\usepackage{cleveref}
\usepackage{xspace, calc}
\usepackage{pgfplots}
\usepackage[utf8]{inputenc}

\usepackage{todonotes}

\newcommand{\new}[1]{{#1}}

\newtheorem{theorem}{Theorem}

\newcommand{\OPT}{\ensuremath{\mathit{OPT}}\xspace}
\newcommand{\weight}[1]{\ensuremath{\mathcal{L}_{#1}}\xspace}

\newcommand{\smin}{\ensuremath{s_{min}}\xspace}
\newcommand{\smax}{\ensuremath{s_{max}}\xspace}

\newcommand{\mapping}{\textsc{SpeedScaling}\xspace} %{\textsc{CriticalPath}\xspace}
\newcommand{\scheduling}{\textsc{Speeds\&Scheduling}\xspace}

\newcommand{\spalgo}{\textsc{SPG-speed}\xspace}
\newcommand{\cvxmapping}{\textsc{CVX-speed}\xspace}

\newcommand{\cvxsched}{\textsc{APX-sched}\xspace}

\newcommand{\roundmap}{\textsc{APX-D-speed}\xspace}
\newcommand{\ilpmap}{\textsc{ILP-D-speed}\xspace}

\newcommand{\roundsched}{\textsc{APX-D-sched}\xspace}
\newcommand{\ilpsched}{\textsc{ILP-D-sched}\xspace}

\title{Energy Minimization in DAG Scheduling on MPSoCs at Run-Time: Theory and Practice}
\date{}

\usepackage{authblk}
\author[1]{Bertrand Simon}
\author[2]{Joachim Falk}
\author[1]{Nicole Megow}
\author[2]{Jürgen Teich}
\affil[1]{Universität Bremen}
\affil[ ]{\textnormal{\texttt{\{bsimon, nicole.megow\}@uni-bremen.de}}}
\affil[2]{Friedrich-Alexander-Universität Erlangen-Nürnberg (FAU)}
\affil[ ]{\textnormal{\texttt{\{joachim.falk, juergen.teich\}@fau.de}}}
\date{}                     %% if you don't need date to appear
\begin{document}

	\title{%Algorithms for 
	Energy Minimization in DAG Scheduling on MPSoCs at Run-Time: Theory and Practice
		%Some Cool Title (math optimization, latency guarantee, power consumption, speed scaling, new and old results)\\
		%\thanks{Funded by the Deutsche Forschungsgemeinschaft (DFG, German 
		%	Research Foundation) --– Project Number 146371743 --- TRR 89 Invasive Computing.}
	}

	\maketitle
	
	\begin{abstract}
Static (offline) techniques for mapping applications given by task
graphs to MPSoC systems often deliver overly pessimistic and thus
suboptimal results w.r.t. exploiting time slack in order to minimize
the energy consumption. This holds true in particular in case
computation times of tasks may be workload-dependent and becoming
known only at runtime or in case of conditionally executed tasks or
scenarios. This paper studies and quantitatively evaluates different
classes of algorithms for scheduling periodic applications given by
task graphs (i.e., DAGs) with precedence constraints and a global
deadline on homogeneous MPSoCs purely at runtime on a per-instance
base. We present and analyze algorithms providing provably optimal
results as well as approximation algorithms with proven guarantees on
the achieved energy savings. For problem instances taken from
realistic embedded system benchmarks as well as synthetic scalable
problems, we provide results on the computation time and quality of
each algorithm to perform a) scheduling and b) voltage/speed
assignments for each task at runtime. In our portfolio, we distinguish
as well continuous and discrete speed (e.g., DVFS-related) assignment
problems. In summary, the presented ties between theory (algorithmic
complexity and optimality) and execution time analysis deliver
important insights on the practical usability of the presented
algorithms for runtime optimization of task scheduling and speed
assignment on MPSoCs.
	\end{abstract}

	\newpage
	
	%--------------------------------------------------------------------
	\section{Introduction}
	%--------------------------------------------------------------------	
	
	Dynamic voltage and frequency scaling (DVFS) on modern processors is a means to actively control the power and energy consumption of an MPSoC. It is used for 
    thermal chip management in combination with dynamic power management (DPM) \cite{Chen0K14}. But it can also be used in the context of dynamic energy minimization of 
    programs executed on the MPSoC, e.g., for real-time applications. Here, a plethora of methods has been proposed to optimize the mapping (including task assignment and scheduling) 
    of tasks of one or multiple applications to processor cores including the selection of processor speed(s) such that, given worst-case task execution times, a global deadline is met.  
    %Yao et al.~\cite{yaoDS95} initiated the algorithmic study of speed scaling in 1995. This area received a lot of attention since then; see the surveys~\cite{albers10,iraniP05}.
    %Whereas first approaches were considering 
    \new{While first investigations only considered} uni-processor systems, a great number of approaches has emerged to apply DVFS optimization algorithms offline when 
    targeting MPSoCs \cite{LangenJ09,Li08,NelsonMMSNG11}. These approaches, however, generally suffer from assuming fixed execution times of tasks given (e.g., WCETs). 
    However, for most applications, the execution times of tasks may depend on the workload to be processed. Or, tasks may only be conditionally executed
    according to control flow information \cite{ShinK03}. Hence, a static assignment of schedule and speeds for executing tasks might not be optimal.
    Choudhury et al. \cite{ChoudhuryCK07} proposed a combination of offline techniques to compute worst-case and average case execution times of tasks. At run-time, 
    a computationally inexpensive method calculates observed slack, and adapts processor speeds for energy reduction,
    while still guaranteeing a global deadline not to be violated. Other approaches such as \cite{SinghDK13} exploit the knowledge of special 
    models of computation such as synchronous dataflow (SDF) to apply a mixed offline and online DVFS optimization for MPSoCs. 
    Still, the structure of the task graph and thus periodicity of executions is assumed static. 
    In most general applications, however, both the execution times and the task graph structure may vary over time. 
    Here, approaches using control/dataflow graphs (CDFGs) have been proposed such as in the work of Tariq et al. \cite{TariqW17}.
    However, the presented computationally complex analysis and optimization is again purely static as task execution probabilities are used
    and thus only the \emph{expected} energy consumption is targeted. 
   
    \new{On the theoretical side, Yao et al.~\cite{yaoDS95} initiated the algorithmic study of speed scaling in 1995. This area received a lot of attention since then; see the surveys~\cite{albers10,iraniP05}.} %In algorithm theory, after Yao et al.~\cite{yaoDS95} initiated the algorithmic study of speed scaling, 
    %Many algorithms have been designed to produce energy-efficient schedules in various
     %contexts, see, e.g., \cite{albers10,MegowV18}. 
     \new{Most of these studies focus on scheduling {\em independent tasks} (without precedences) and a {\em single processor}. Regarding the speed choice model, only few theoretical works address the {\em discrete speed} model which is computationally much more complex but more realistic; see, e.g.~\cite{KwonK05,LiY05,MegowV18}.}
     %see for instance~\cite{MegowV18} among the most recent ones, and are therefore not directly related to this paper. 
     Most related to our investigations is the work by Aupy et al.~\cite{aupy} that studies the
     problem of minimizing the energy consumption under a given mapping of
     tasks to cores, and where the power consumed by a core
     running at speed $s$ is equal to $s^\alpha$. They consider both the
     continuous and the discrete speed model. Pruhs et
     al.~\cite{pruhs2008speed} focus on the problem of minimizing the
     makespan under an energy budget in the continuous speed model with the same power law. In this framework, they designed an
     approximation algorithm with a polylogarithmic ratio. Bampis et
     al.~\cite{bampis} later proposed a 2-approximation for the same
     problem, which matches the best known algorithm for
     makespan minimization without energy considerations~\cite{graham}. Our
     contributions include to rephrase these results for our framework
     (energy minimization with a fixed deadline) and analyze the
     algorithms performance experimentally. \new{We also add new results building on this earlier work.}
     
    A major goal of this paper is to analyze whether algorithms providing provably optimal results 
    or at least approximation bounds on the quality of the results 
    can be implemented and practically applied in a real MPSoC system to be executed at runtime.
    In this regard, at least to our knowledge, the following questions have so far not satisfactorily been answered for the problem of scheduling task graphs
    purely at runtime based on dynamically emerging task dependence structure and worst case execution times.
    \medskip
    
    \begin{itemize}
     \item {\bf Are there in theory sound algorithms that also can be applied in practice on an MPSoC?}
           E.g., depending on the absolute time scale, many real-world applications do require solutions to be computed within a time scale of 1 to 10 ms
           in order to be of practical use.
     \item {\bf How do these execution times scale with the problem size?}
           Problem instances ranging in size between 10 to 500 tasks should %be able to 
           be handled in practice within such 
           time scales. If not possible: 
     \item {\bf Are there fast and scalable  algorithms with provable approximation bounds on the optimality of energy consumption?}
    \end{itemize}
	
	%\todo[inline]{Some motivation for both models: critical path (mapping-given) and scheduling. Latter seems clear.}
	%\todo[inline]{Can we add in motivation: super fast execution times for the algorithms? milliseconds? by complexity this means no optimal solution. still guarantees needed = approx.}
	%\todo[inline]{There are claims about experiments results which will need to be checked afterwards}
		
	Our main contribution is to bring together theoretical and computational results for continuous and discrete speed scaling for precedence-constrained task systems with the goal to minimize the energy consumption.
	We distinguish two classes of problems: The first assuming a processor assignment and schedule of tasks on cores given, and the second computing the full schedule including the task to processor mapping as well while minimizing energy.
	We present algorithms building on mathematical optimization techniques such as convex and integer linear programs as well as rounding solutions of relaxations. Previously known methods are adapted to our setting and we also provide new results. For the full portfolio of considered settings (continuous/discrete speed choice, unlimited/bounded number of processors) we give both, an exact algorithm and a computationally efficient (i.e., polynomial-time) algorithm. In cases where optimal polynomial-time algorithms are ruled out under standard complexity assumptions, we give a polynomial-time algorithm with an approximation guarantee, i.e., we guarantee for any input instance that the total energy needed by our algorithm to finish all tasks by the deadline is at most a certain factor away from a minimum energy needed by any algorithm.  
	Such theoretical approximation results give very rigorous worst-case guarantees on the solution quality under any possible input. They are of high importance particularly for safety-critical real-time applications. 
	In our experiments on real-world instances, it is shown that the solution quality is substantially better than the ones guaranteed in our theorems for the worst case. 
		
	Moreover, we rigorously analyze the applicability of all of our algorithms on problem instances taken from realistic embedded system benchmarks as well as synthetic scalable problems.
	As one result, it turns out that the mathematical optimization methods are applicable for MPSoC system applications despite their complexity. Running times between 1 to 10 $ms$ for instances up to 100 tasks are in the acceptable range for many applications.
	If not, also a linear-time algorithm (previously used in similar settings  \cite{prasmus,europar,aupy}) that combines optimality with scalable performance for a majority of task graph instances exhibiting a series/parallel dependence structure is presented. 
	Overall, our results include new and old algorithms with optimality/approximation guarantees while revealing their practicability for use in MPSoCs. 
		
	%{\bf Algorithmic related work.} A
		
	%--------------------------------------------------------------------	
	\section{Formal problem definition and notation}
	%--------------------------------------------------------------------
	
	%We consider the problem of scheduling a set of tasks on $m$ 
	%identical cores -- not if we set speeds individually
	%cores without preemption. 
	We are given a set of tasks to be executed without preemption on $m$ cores. Precedence relations between the tasks are given as a directed acyclic graph $G=(V,E)$, where each node in the graph is associated with a task. If there is an arc in $E$ from task $j$ to task $k$ then task $k$ cannot start before task $j$ is completed. A task $j\in V$ has a nominal execution time, or weight, $w_j\geq 0$. 
	
	For comparability of the analyzed algorithms, we assume a homogeneous multi-processor architecture in the following with uniform cores. 
	At any time, the speed $s$ of a core can be set to any {\em eligible} value between $\smin> 0$ and $\smax\geq \smin$, and it is part of a scheduling algorithms decision to which speed to set the processor. It depends on the particular model, which values in $[\smin, \smax]$ are eligible; we consider the {\em continuous model}, in which any rational value is eligible, and the {\em discrete model}, which allows speeds only from a given finite set of speeds. A core that is set to speed $s$ consumes power at the rate $s^\alpha$, where $\alpha\geq 1$ is a small constant. The total energy consumed is the power consumption integrated over time. 
	
	In the continuous model, we may assume that a task is executed at a uniform speed. This follows directly from the convexity of the power function~\cite{yaoDS95}. 
	For discrete speeds, we add the restriction that a task has to run at a uniform speed. This is a reasonable assumption as in many processing environments it is not possible to change the processor speed during the execution of a task.
	%\nic{What about discrete speeds, here this is not true. We may consider two discrete variants: one where a task can run at several speeds, and one where a task is restricted to run at uniform speed.
	%} 
	If a task $j$ of weight $w_j$ is executed at speed $s_j\in[\smin,\smax]$, then the time to complete is $x_j=w_j/s_j$ and the energy consumed during the computation of $j$ is 
	\begin{equation}
	E_j=x_j\cdot s_j^\alpha= w_j\cdot s_j^{\alpha-1}=w_j^\alpha/x_j^{\alpha-1}.\label{eq:energy}
	\end{equation}
	
	We consider the following problem: given a deadline $D\!>\!0$ and a node-weighted graph $G=(V,E,w)$, schedule all tasks in graph $G$ and decide upon the processor speeds such that all tasks finish  before the deadline $D$ and the total energy consumption is minimized. If minimizing the energy consumption is intractable, we design approximation algorithms. An algorithm is called an $r$-approximation if it always computes a solution finishing before the deadline, with an energy consumption being at most $r$ times the minimal energy consumption.

	In our investigations we distinguish two problem classes of different complexity:	
	\begin{itemize}
		\item \mapping: we are given the mapping of each task to its core and the order in which each core executes the tasks mapped to it (encoded in $G$). %, which is already included in $G$. 
		The problem is then equivalent to minimizing the \emph{critical path} of the graph $G$. That is, find speeds such that the total execution time of the longest path (w.r.t.\ execution times $x_j$) is minimized. %, i.e., minimizing the length of the longest path in the graph.
		\item \scheduling: in addition to selecting the speeds at which each task should be executed, we provide a schedule for the tasks, i.e., we determine the core and the starting time for each task. %, in order to minimize the \emph{makespan} of the schedule (the completion time of the last task).
	\end{itemize}

	%--------------------------------------------------------------------	
	\section{Continuous speeds}
	%--------------------------------------------------------------------
	
	We consider the setting in which each core can be set to any rational speed value in the given interval $[\smin,\smax]$.
	
	\subsection{\mapping Problem}
	
	As mentioned earlier, this problem is equivalent to determining the speeds such as to minimize the \emph{critical path} of the graph $G$. This problem has been studied to some extent before. We summarize relevant known algorithms and provide new ones. We present two algorithms: 
	\begin{enumerate}
		\item an optimal polynomial time algorithm \cvxmapping which relies on a convex programming formulation inspired by the idea of Bampis et al.~\cite{bampis};
		\item a linear-time algorithm \spalgo for a special graph class, namely Series-Parallel Graphs, which are very common in practice. Our algorithm is a small modification of an algorithm in  \cite{prasmus,europar,aupy} and it computes an exact optimal solution when there is no limitation in the speeds. Our experiments show that this limitation is not prohibitive in our context. %, see Section~\ref{sec:experiments};
%		\item a linear programming formulation \lpestim relaxing the power consumption, which can be solved very fast and performs well if the deadline is rather tight.
		%relies on an estimation of the energy in order to model the problem as a linear program. Its running time is then faster than \cvxmapping, but its performance can be poor.
	\end{enumerate}
	
	Details on the algorithms follow below. The experimental evaluation is presented in Section~\ref{sec:experiments}.	
	
	\subsubsection{\cvxmapping}

	%\begin{theorem}
	We provide a convex programming formulation with linear constraints that computes the exact solution for the energy minimization problem in the \mapping setting. Such programs can be solved in polynomial time up to an arbitrary precision~\cite{nesterovN94} with the Ellipsoid method. The formulation is inspired by a convex program for makespan minimization by Bampis et al.~\cite{bampis}.
	%\end{theorem}
	
	Each task $j$ is associated to a constant speed $s_j$. The variable $x_j$ represents the processing time of Task $j$ in the solution, which is equal to $w_j/s_j$. The variable $d_j$ represents the completion time of task $j$. %, or more precisely, is at least equal to this completion time. 
	\begin{align}
	\text{min } \sum_{j\in V}  & \frac{w_j^\alpha}{x_j^{\alpha-1}} \label{eq:cvx1} \\
	\text{s.t.\quad} d_j&\leq D, & \forall j\in V\\
	x_j &\leq d_j, & \forall j \in V\\
	d_j+x_k &\leq d_k, & \forall (j,k)\in E\\
	w_j/\smax \leq x_j & \leq w_j/\smin, & \forall j\in V. \label{eq:cvx2}
	\end{align}
	
	The first three constraints ensure that tasks are executed one after the other, without preemption, respecting the precedence constraints and meeting the deadline $D$. Constraint~\ref{eq:cvx2}  ensures  that the speed limits are respected. Finally, the objective function computes the energy consumption for the schedule that is to be minimized. For a computed solution of the convex program, the speed $s_j$ for task $j$ is implied by $w_j/x_j$. We therefore have the following result.
	
	\begin{theorem}
		\cvxmapping computes an optimal solution in polynomial time.
	\end{theorem}
	
	\subsubsection{\spalgo} 
	
	In the most general definition by Lawler~\cite{lawler1978sequencing}, series-parallel graphs (or SP-graphs) %can represent slightly different classes of graphs, as some definitions used in the literature as more restrictive than others. 
	%We use here the general definition of~\cite{lawler1978sequencing}, in which an SP-graph is 
	are defined recursively as being either a single task, the series composition of two graphs (noted $(G_1;G_2)$), or the parallel composition of two graphs (noted $(G_1||G_2)$). In $(G_1;G_2)$, the tasks of $G_2$ cannot start before all tasks of $G_1$ have terminated. In $(G_1||G_2)$, there exist no precedence constraints between the tasks of $G_1$ and~$G_2$.
	
	In the context of minimizing the makespan of malleable jobs, an algorithm has been proposed and studied in \cite{prasmus,europar}, and a similar algorithm has been used in our context in \cite{aupy}. The principle of the algorithm is to define an \emph{equivalent} task of a series and a parallel composition of two graphs. Specifically, if $\weight{G}$ represents the equivalent weight of $G$, we have:
	
	\begin{itemize}
		\item $\weight{T_i} = w_i$
		\item $\weight{(G_1;G_2)} = \weight{G_1} + \weight{G_2}$
		\item $\weight{(G_1||G_2)}^{\alpha} = \weight{G_1}^\alpha + \weight{G_2}^\alpha$
	\end{itemize}
	
	The problem of selecting the speeds for a graph $G$ in order to minimize the energy consumption is then equivalent to the problem of selecting the speed for a unique task of weight \weight{G}. The minimum energy necessary to schedule a graph $G$ under a deadline $D$ is therefore equal to $\weight{G}^\alpha/D^{\alpha-1}$, using the speed $\weight{G}/D$, see \Cref{eq:energy}. %, the proof can be found in \cite{prasmus,europar,aupy}.
	In order to compute the speed at which each task has to be scheduled in such a solution, the algorithm \spalgo associates a speed $s$ to each subgraph:
	
	\begin{itemize}
		\item $s(G) = \weight{G}/D$
		\item In $(G_1;G_2)$, $s(G_1)=s(G_2) = s(G_1;G_2)$.
		\item In $(G_1||G_2)$, $s(G_1) = s(G_1||G_2) \weight{G_1} / \weight{(G_1||G_2)}$. 
	\end{itemize}
	
		This result however requires to use speeds arbitrarily large, so the solution found may not respect the speed bounds, as specified in the following theorem.
	
	\begin{theorem}[\!\cite{prasmus,europar,aupy}]
		Given an SP-graph and ignoring the constraints $\smin$ and $\smax$, \spalgo computes an optimal solution in linear time.
	\end{theorem}

%	\subsubsection{\lpestim} \hfill\\
%	
%	
%	
%	\todo[inline]{(BS) Detail this formulation. The linear approximation is relevant in some cases, but there are simple examples in which this approach gives arbitrarily bad results (e.g., with a loose deadline).}

	\subsection{\scheduling Problem}
	\label{subsec:cont_sched}
	
	%\nic{This algorithm is not implemented by us. As described below, it takes solving the convex program (already implemented) and running a simple algorithm.}

	Consider the setting in which an algorithm determines both, the speed
	allocation and the actual schedule including the mapping of tasks to
	cores. If the optimal solution requires to use the speed $\smax$ for
	each task, then computing a schedule meeting a given deadline is
	already an NP-hard problem, as it is reducible to the classic
	$P|\mathit{prec}|C_{max}$ problem in the Graham three-field notation. The \scheduling problem can
	therefore not have an approximation algorithm unless $P=NP$, as this includes computing a schedule meeting the given deadline. The best known
	scheduling algorithm for $P|\mathit{prec}|C_{max}$ is a
	2-approximation\cite{graham}, and cannot be improved under some
	complexity assumptions~\cite{Svensson11}. We therefore assume that the
	optimal solution uses speeds at most $\smax/2$, in order to focus on
	the problem of minimizing the energy and not on meeting the deadline,
	which is not the core of this paper.
		We show the following result.
	
	%We no longer consider that the mapping of tasks to processors is given. 
	
	\begin{theorem}\label{thm:apx-continuous-with-scheduling}
		\cvxsched is a $2^{\alpha-1}$-approximation if the optimal solution uses speeds at most $\smax/2$.
	\end{theorem}
	
	%\todo[inline]{(BS) I will give more details, but to summarize: we can only compare the performance to a schedule finishing in time D/2, as scheduling algorithms loose a factor 2 in the makespan. For instance, if the deadline is so tight that the shortest schedule is close to the deadline, meeting the deadline is already NP-hard, even without energy considerations.}
	
	The main idea of the algorithm builds on work in~\cite{bampis} for the related problem of minimizing the makespan under a fixed energy budget. The algorithm consists of two steps: firstly, a convex program is solved for computing the optimal speeds in a particular relaxation. Secondly, we fix these speeds and run a greedy heuristic for assigning the tasks to cores.
	%This algorithm is almost identical to the one presented in~\cite{bampis} for the related problem of minimizing the makespan under a fixed energy budget. The first step consists in solving the convex program of \cvxmapping with one additional constraint:
	The convex programming relaxation is as follows (recall that $m$ is the number of cores).
	\begin{align}
	\text{min } \  \sum_{j\in V} \frac{w_j^\alpha}{x_j^{\alpha-1}}\\
	\text{s.t. } \ \sum_{j\in V} x_j /m &\leq D/2\\
	d_j &\leq D/2, & \forall j\in V\\
	x_j &\leq d_j, & \forall j \in V\\
	d_j+x_k &\leq d_k, & \forall (j,k)\in E\\
	w_j/\smax \leq x_j & \leq w_j/\smin, &\forall j\in V.
	\end{align}
	
	Given an optimal solution for this program, we fix the speeds for the tasks. In the second step of the algorithm, we schedule the tasks using a list scheduling algorithm proposed by Graham~\cite{graham}. That is, we consider tasks in any topological ordering (i.e., respecting the given precedence order) and assign a task to the core with currently smallest last completion time. If the makespan $C$ obtained is smaller than $D$, the speeds are then lowered by a factor $C/D$.
	
	\begin{proof}[Proof of Theorem \ref{thm:apx-continuous-with-scheduling}]
		For a fixed speed assignment let $V:=\sum_{i\in V} x_i /m$ denote the {\em volume} and let $L$ denote the length of the critical path in $G$. Both, volume and critical path, are well known lower bounds on the makespan. Graham's list scheduling~\cite{graham} yields a makespan of at most $V+L$. %
		The convex program computes a speed assignment that minimizes the energy among all speed assignments for which both the volume and the critical path are not larger than $D/2$. Hence, Graham's list scheduling achieves a schedule where all tasks complete before $V+L \leq D$ and, thus, all tasks meet the deadline.
		
		On the other hand, one can show that the energy consumed by this schedule is at most a factor $2^{\alpha-1}$ larger than the optimal. Indeed, consider an optimal schedule of makespan~$D$ using speeds at most $\smax/2$, and multiply every speed by $2$. We obtain a speed assignment which is a solution to the convex program above, and whose energy cost is a factor $2^{\alpha-1}$ away from the optimal. As the speed assignment computed by the algorithm minimizes the objective function, its energy cost is not larger.
	\end{proof}
	
	In Section~\ref{sec:experiments}, we will show that on real-world instances, the solution quality is substantially better than the one guaranteed in Theorem~\ref{thm:apx-continuous-with-scheduling} above. %We emphasize that the theorem gives a very rigorous guarantee on the performance under any possible input. Such strong guarantees are of high importance particularly for safety-critical applications. 
	Finally, we remark that the problem is computationally highly intractable. Even for a given speed assignment, it is NP-complete to compute an optimal schedule even if all tasks have unit execution time~\cite{Ullman75} or if there are no precedence relations~\cite{GareyJohnson78}. %Under some complexity assumption, it is even NP-hard to compute a solution with a makespan that is within a factor $2-\eps$, for any $\eps>0$ \cite{svensson}. Thus, we cannot hope for an optimal solution (or even a $(2-\epsilon)$-approximation in polynomial time.... translates to tight energy result???	

	%----------------------------------------------------------------------------------	
	\section{Discrete speeds}
	%----------------------------------------------------------------------------------
	
	%\nic{These algorithms are not implemented by us. As described below, it takes solving a linear/convex program associated with a simple algorithm.}
	
	Consider the setting in which each core can run at $k\in \mathbb{N}$ possible speeds $v_1,v_2,\dots, v_k$ with $v_i< v_{i+1}$. Let the maximum ratio of speeds be $r = \max_i v_{i+1}/v_i$. Note that the mapping problem in this setting is already NP-hard even with $k=2$~\cite{aupy}. However, the more general model in which speed modifications are allowed during the execution admits a polynomial exact algorithm~\cite{aupy}. We also underline that the approximation ratios given in this section still hold if the optimal solution is allowed to use any rational speed in the interval $[v_1;v_k]$.
	
	\subsection{\mapping problem}
	
	Assume the task-to-core assignment is given and we need to determine the speeds such as to minimize the \emph{critical path} of the graph $G$. %  
	We present two algorithms: (1) an optimal exponential time algorithm \ilpmap based on an integer linear programming (ILP) formulation, (2) a polynomial time algorithm \roundmap that solves a convex program within an approximation factor $r^{\alpha-1}$. 
	
	%\begin{itemize}
	%	\item \roundmap relies on a convex program to approximate the problem to a factor $r^{\alpha-1}$ in polynomial time.
	%		\item \ilpmap solves the problem optimally in exponential time using an ILP formulation.
	%\end{itemize}
	
	\subsubsection{\ilpmap}
	
	We define $nk$ boolean variables $y_{i,\ell}$ which are equal to 1 if task $i$ runs at speed $v_\ell$ and to 0 otherwise, and consider the following program similar to the convex program \eqref{eq:cvx1}-\eqref{eq:cvx2}. The main difference is that the execution time of a task $i$ is now equal to  $\sum_{\ell\leq k}\frac{w_i}{v_\ell}y_{i,\ell}$ and its energy consumption is equal to $\sum_{\ell\leq k} w_i {v_\ell^{\alpha-1}} y_{i,\ell}$.
	\begin{align}
	\text{minimize }  \sum_{i\in V} {w_i} &\sum_{\ell\leq k} {v_\ell^{\alpha-1}} y_{i,\ell} \label{ilp:1}\\
	d_i&\leq D &\forall i\in V\\
	\sum_{\ell\leq k}\frac{w_i}{v_\ell}y_{i,\ell} &\leq d_i &\forall i \in V\\
	d_i + \sum_{\ell\leq k}\frac{w_j}{v_\ell}y_{j,\ell} &\leq d_j &	\forall (i,j)\in E
	\end{align}
	\begin{align}
	\sum_{\ell\leq k} y_{i,\ell} &= 1 & \forall i\in V\\
	\forall \ell \leq k ~~y_{i,\ell}&\in\{0,1\} & \forall i\in V. \label{ilp:2}
	\end{align}
	
	The correctness of this ILP formulation therefore follows from the correctness of \cvxmapping.
	
	\begin{theorem}
		\ilpmap computes an optimal solution in exponential time.
	\end{theorem}

	In general, integer linear programs cannot be solved in polynomial time. However, our experiments show that on the datasets considered (up to 1000 tasks), this algorithm is at most 5 times slower than the polynomial-time algorithm \cvxmapping.
	
	\subsubsection{\roundmap}
	
	The following algorithm is inspired by \cite{aupy}. In a first step, we 
	compute optimal {\em continuous} speeds $\bar s_j$ for each task $j$. This is done by running the fast algorithm \spalgo, and, in case this algorithm does not succeed ({\it e.g.,} the SP-graph restriction is not met), we solve the convex program \eqref{eq:cvx1}-\eqref{eq:cvx2} (algorithm \cvxmapping) with $\smin=v_1$ and $\smax=v_k$. Then, the we run each task $j$ at the speed $s_j$ that is equal to the smallest speed $v_i$ such that $v_i\geq \bar s_j$.
	
	\begin{theorem}
		\roundmap computes an $r^{\alpha-1}$-approximate solution in polynomial time.
	\end{theorem}	
	
	\begin{proof}
		%We now show that this algorithm is a $r^{\alpha-1}$-approximation. 
		Consider a speed setting computed by the algorithm. Observe that the tasks respect the deadlines as the speeds $s_j$ are not smaller than the speeds $\bar s_j$ that gave a valid solution. Let $\OPT$ be the energy consumed in an optimal solution.
		First, note that the energy consumed by executing each task at speed $\bar s_j$ is not larger than $\OPT$. 
		The algorithm runs each task $j$ at speed $s_j$, consuming an energy $w_js_j^{\alpha-1}$. The total energy consumed is then:
		$$E = \sum_j w_js_j^{\alpha-1} \leq \left(\frac{s_j}{\bar{s_j}}\right)^{\alpha-1} \sum_j w_j \bar s_j ^{\alpha-1} \leq r^{\alpha-1} \OPT. \qedhere$$		
	\end{proof}

	\subsection{\scheduling problem}
	
	In this setting, an algorithm determines both, the speed allocation and the actual schedule including the mapping of tasks to cores.	
	We present two algorithms:
	(1) an optimal exponential time algorithm \ilpsched based on solving an ILP, (2) a polynomial time algorithm \roundsched that solves a convex program within approximation factor $(2r)^{\alpha-1}$. 
	
	%\begin{itemize}
	%	\item \roundsched relies on a convex program to approximate the problem to a factor $(2r)^{\alpha-1}$ in polynomial time under some conditions.
	%	\item \ilpsched solves the problem optimally in exponential time using an ILP formulation.
	%\end{itemize}
	
	\subsubsection{\ilpsched}
	
	We extend %\ilpmap 
	the ILP \eqref{ilp:1}-\eqref{ilp:2} by adding $nm$ boolean variables $z_{i,c}$ equal to $1$ if task $i$ is executed on core $c$ and to $0$ otherwise, as well as $n^2$ variables $e_{i,j}$ indicating if task $i$ has to be scheduled before task $j$. %equal to $1$ if task $i$ has to be scheduled before task $j$ and to $0$ otherwise. 
	In particular, if two tasks are executed on the same core, then either $e_{i,j}$ or $e_{j,i}$ equals $1$.

	\newlength{\myl}
	\settowidth{\myl}{$\forall \ell \leq k,~$}
	\begin{align}
	\text{minimize }  \sum_{i\in V} {w_i} &\sum_{\ell\leq k} {v_\ell^{\alpha-1}} y_{i,\ell} \\
	d_i&\leq D &\forall i\in V\\
	\sum_{\ell\leq k}\frac{w_i}{v_\ell}y_{i,\ell} &\leq d_i & \forall i \in V\\
	d_i + \sum_{\ell\leq k}\frac{w_j}{v_\ell}y_{j,\ell} \leq d_j &+D(1-e_{i,j}) \label{eq:edge} & \forall i,j\in V\\
	\sum_{\ell\leq k} v_{i,\ell} &= 1 & \forall i\in V\\
	v_{i,\ell}&\in\{0,1\} & \hspace{-\myl} \forall \ell \leq k,~ \forall i\in V\\
	\sum_{c\leq m} z_{i,c} &= 1 & \forall i\in V 
	\end{align}
	\begin{align}
	\settowidth{\myl}{$\forall i\in V,~$}
	z_{i,c}&\in\{0,1\} & \forall i\in V, ~\forall c \leq m\\
	e_{i,j}&\in\{0,1\} &\forall i,j \in V\\
	e_{i,j} &=1 & \forall (i,j) \in E\\
	%	\forall i_1,i_2,i_3 \in Ee_{i_1,i_2} + e_{i_2,i_3} - e_{i_1,i_3} &\leq 1\\
	\settowidth{\myl}{$\forall i,j\in V,~$}
	z_{i,c} + z_{j,c} - e_{i,j} - e_{j,i} &\leq 1 \label{eq:core}  &\hspace{-\myl}\forall i,j\in V,~ \forall c\leq m
	\end{align}
	
	\Cref{eq:edge} ensures that task $j$ is executed after task $i$ if
	$e_{i,j}=1$, and does not have any impact if $e_{i,j}=0$, so the
	program returns the same result as \ilpmap on a graph where the edges
	are represented by the variables $e_{i,j}$. The second important
	constraint is \Cref{eq:core}, which ensures that if two tasks
	belong to the same core, either $e_{i,j}$ or $e_{j,i}$ equals~$1$. Therefore, a valid valuation of the variables $e_{i,j}$ corresponds to a directed graph which contains the edges of $E$, and which contains an edge between any two tasks that are placed on the same core (by the variables $z_{i,c}$). This corresponds to a valid input to the \ilpmap programming, so we have the following result.
	
	\begin{theorem}
		\ilpsched computes an optimal solution in exponential time.
	\end{theorem}

	\subsubsection{\roundsched}
	
	This algorithm combines the ideas of \cvxsched and \roundmap: assuming again that the optimal solution uses speeds at most $v_k/2$, solve the convex program of \cvxsched in order to associate each task to a speed $\bar s_j\in[v_1;v_k]$. Then, the algorithm runs each task $j$ to the speed $s_j$ equal to the smallest speed $v_i$ such that $v_i\geq \bar s_j$, and schedules the tasks using a list scheduling algorithm.
	
	\begin{theorem}\label{roundtheo}
		\roundsched computes a $(2r)^{\alpha-1}$-approximate solution in polynomial time if the optimal solution uses speeds at most $v_k/2$.
	\end{theorem}
	
	\begin{proof}
		We first note that, similarly to the \roundmap case, the energy used by the schedule obtained by \roundsched is at most a factor $r^{\alpha-1}$ away from the energy used by the \cvxsched solution. Then, assuming that the optimal solution uses speeds at most $v_k/2$, we know that the energy used by the \cvxsched solution is within a factor $2^{\alpha-1}$ of the optimal energy consumption. Combining these two results completes the proof.		
	\end{proof}
	%--------------------------------------------------------------------
	\section{Experimental Results}\label{sec:experiments}
	%--------------------------------------------------------------------

        In order to evaluate the quality of the presented approaches, we use a total of $5 \times 5$ benchmark graphs, i.e., five groups of five graphs of similar size.
        Our 5 smallest graphs are comprised of around 10 tasks and are derived from the Embedded System Synthesis Benchmarks Suite (E3S)~\cite{e3s}. These instances target processors of maximum frequency 250MHz, with a minimum frequency equal to 0.1MHz. 20 eligible speeds 
        can be selected equally distributed between these limits. The deadlines associated to these graphs equal a few milliseconds, and are rather tight: several tasks need to be run at the maximum frequency.
        For larger graphs with 50, 100, 500, and 1000 tasks each, we selected graphs from the GENOME dataset of the Pegasus library~\cite{Pegasus}. The homogeneous processors used here were specified 
        at a maximall frequency equal to 1.0GHZ and again 20 equidistant speed setting, but assumed looser deadlines. All benchmarks belong to the class of SP-graphs, thus allowing the application of \spalgo.

        The benchmarks are executed on an Intel(R) Core(TM) i7-4770 CPU running at 3.40\,GHz with 32\,GiB of RAM using Ubuntu 18.04 LTS as underlying OS.
        To solve the ILPs for the \ilpsched and \ilpmap approaches, we use CPLEX 12.6 with a running time deadline of 5s.
        For the convex programs used by the \cvxmapping and \roundmap approaches, we used MOSEK 8.1.
        
        \Cref{fig:results-mapping-given} presents the results for the \mapping problem, both in the continuous (\cvxmapping and \spalgo) and discrete speed (\ilpmap and \roundmap) settings.
        \definecolor{clr1}{RGB}{160,25,28}
        \definecolor{clr2}{RGB}{210,140,60}
        \definecolor{clr3}{RGB}{43,131,186}
        \definecolor{clr4}{RGB}{171,221,164}
        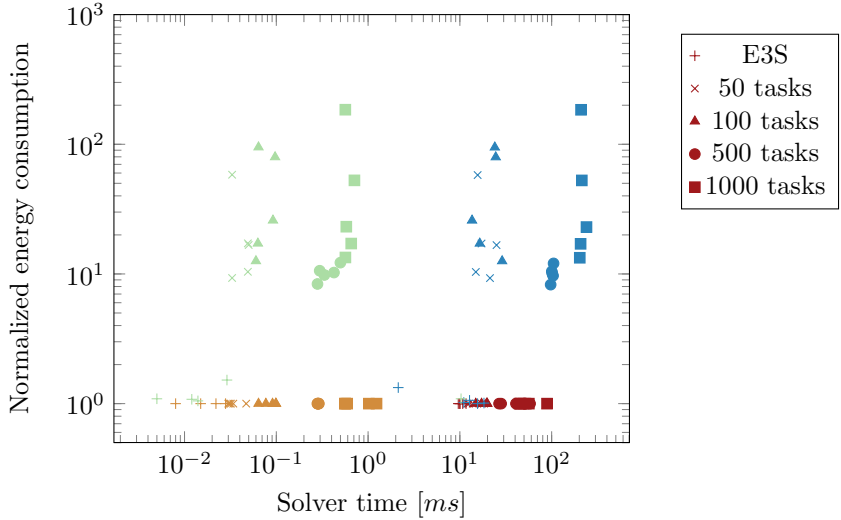
\begin{figure}[t]
          \centering
          \begin{tikzpicture}
            \begin{axis}[%
              xmode=log,
              ymode=log,
              ymax = 1000,
              xlabel={Solver time [$ms$]},
              ylabel={Normalized energy consumption},
              legend style={at={(1.1,0.75)}, anchor=west},
              legend columns=1,
              scatter/classes={%
                a={mark=+},%
                b={mark=x},%
                c={mark=triangle*},
                d={mark=*},
                e={mark=square*}
                }
            ]
            
              \addplot[color = clr1, scatter, only marks, scatter src=explicit symbolic] table[col sep=comma, meta=mark, x=IIIACVXmappingSolverTime,     y=IIIACVXmappingEnergy    ]{mapping-given.csv};
              \addplot[clr2,   scatter, only marks, scatter src=explicit symbolic] table[col sep=comma, meta=mark, x=IIIASPalgoSolverTime,         y=IIIASPalgoEnergy        ]{mapping-given.csv};
              \addplot[clr3,    scatter, only marks, scatter src=explicit symbolic] table[col sep=comma, meta=mark, x=IVAILPmappingSolverTime,      y=IVAILPmappingEnergy     ]{mapping-given.csv};
              \addplot[clr4,    scatter, only marks, scatter src=explicit symbolic] table[col sep=comma, meta=mark, x=IVAROUNDINGmappingSolverTime, y=IVAROUNDINGmappingEnergy]{mapping-given.csv};
              \addlegendentry{E3S}
              \addlegendentry{50 tasks}
              \addlegendentry{100 tasks}
              \addlegendentry{500 tasks}
              \addlegendentry{1000 tasks}
            \end{axis}
          \end{tikzpicture}
          \caption{\label{fig:results-mapping-given}Depicted above are the trade-offs between solver run-time and energy used by the solution for the four approaches --
%            {\color{red}$\blacksquare$} \lpestim,
            {\color{clr1}$\blacksquare$} \cvxmapping,
            {\color{clr2}$\blacksquare$} \spalgo,
            {\color{clr3}$\blacksquare$} \ilpmap, and
            {\color{clr4}$\blacksquare$} \roundmap{}
            -- that assume that mapping and scheduling is already given (only speed assignment).
            These trade-offs have been determined for five classes of five benchmark graphs each from the E3S benchmark suite and the Pegasus library. The energy displayed is normalized by the minimal energy consumed with continuous speeds.}
        \end{figure}

        Our first observation from the experiments is that the algorithm \spalgo can be applied to all problem instances computing an optimal solution except 
        for one single E3S graph instance where the prescribed speed limits were not respected.
        Moreover, it is really fast, needing at most 0.02ms for each of the five E3S graphs and 1ms only for the largest graphs with 1000 tasks. 
        It can therefore be applied at runtime even for problems with very small and tight deadlines. As a consequence, the algorithm \roundmap runs at a comparable speed, except for the one instance which is not solved by \spalgo.
        Even solving optimally the convex program (\cvxmapping) is possible in less than 10 ms for the E3S benchmarks, \new{15ms for 100-tasks graphs}, but may be unaffordable for very large graphs (in average 60ms for 1000 tasks).
        When solving the ILP for discrete speeds, the solver time can even increase to 200ms for the largest graphs, but we do not observe an exponential growth for this dataset, contrarily to the worst-case theoretical complexity.
        Surprisingly, the quality of the solution of \roundmap is only a few percents away from the optimal discrete solution (\ilpmap). Therefore, \roundmap can obtain near-optimal results two orders of magnitude faster than by solving the ILP, on SP-graph instances.
        The restriction to the discrete speed model implies a higher increase in energy consumption for the GENOME dataset.
        This can be explained by the fact that the deadlines are looser, so the optimal continuous speeds are lower, and being forced to select a discrete speed incurs higher losses.
        %\bert{I don't know if we should keep this sentence} 
        
        \Cref{fig:results-mapping-optimized} presents the results of the \roundsched algorithm that performs also task-to-core assignment and scheduling apart from speed selection. From the color code, it can be seen that the solver times are roughly equal the ones of the \cvxsched algorithm.
        \definecolor{viridis0}{HTML}{DCE319}
        \definecolor{viridis1}{HTML}{440154}
        \begin{figure}[t]
          \centering
          \begin{tikzpicture}
            \begin{axis}[
%             xlabel=$graphs$,
              x dir=reverse,
              xmin=0, xmax=4.8,
              xtick=\empty,
              xticklabel style={rotate=45, anchor=east,yshift=+10pt,xshift=-10pt},
              extra x ticks={0, 1, 2, 3, 4},
              extra x tick labels={E3S,50 tasks, 100 tasks, 500 tasks, 1000 tasks},
              ylabel=Free cores,
              y label style={rotate=40},
              y dir=reverse,
              ymin=1, ymax=128,
              ytick={1,32,64,96,128},
              zmode=log,
              zlabel={Normalized energy consumption},
              z label style={yshift=+2pt,xshift=+35pt},
              point meta max=150,
              point meta min=0,
              colorbar horizontal,
              colormap={reverse}{
              	indices of colormap={
              		\pgfplotscolormaplastindexof{viridis},...,0 of viridis}
              },
              colorbar style={
                  title=Solver time [$ms$],
                  at={(0,1.2)},
                  anchor=north west,
%                 ylabel=Solver time [$ms$],
%                 xtick={0,25,50,75,100,125,150},
                  xtick distance=25,
%                 yticklabel style={
%                     text width=2.5em,
%                     align=right,
%                     /pgf/number format/.cd,
%                         fixed,
%                         fixed zerofill
%                 }
              },
%             title=A Scatter Plot Example
            ]
%             \addplot3[scatter, only marks, mark=triangle*, point meta=explicit] table[col sep=comma, x=graphClassInstance, y=cores, z=IIIBCVXschedulingEnergy,     meta=IIIBCVXschedulingSolverTime    ]{mapping-optimized.csv};
              \addplot3[scatter, only marks, mark=square*,   point meta=explicit] table[col sep=comma, x=graphClassInstance, y=cores, z=IVBROUNDINGschedulingEnergy, meta=IVBROUNDINGschedulingSolverTime]{mapping-optimized.csv};
              \addplot3[scatter, only marks, mark=+,         color=black        ] table[col sep=comma, x=graphClassInstance, y=cores, z=IVBOptimalEnergy                                                 ]{mapping-optimized.csv};
%             \addplot3[surf, green] table[col sep=comma, x=graphClassInstance, y=cores, z=IVAOptimalEnergy]{mapping-optimized.csv};
              \addplot3[mesh, black, domain=0:4  , samples=5, y domain=1:128, samples y=16]{1};
              \addplot3[mesh, black, domain=4:4.8, samples=2, y domain=1:128, samples y=16]{1};
            \end{axis}
          \end{tikzpicture}
          \caption{\label{fig:results-mapping-optimized}Consumed energy of the $5\times5$ benchmark graphs for solutions found by the \roundsched approach (squares) subject to a fixed number of available (free) cores ranging from 1 to 128. The results are normalized: a value of 1 corresponds to the case with optimum discrete speeds and infinitely many cores.
            The required solver time to find these solutions ranges from 7\,ms ({\color{viridis0}$\blacksquare$}) to 150\,ms ({\color{viridis1}$\blacksquare$}) according to the given color key. The crosses denote optimal-energy solutions as determined 
            by the \ilpsched approach.}
        \end{figure}
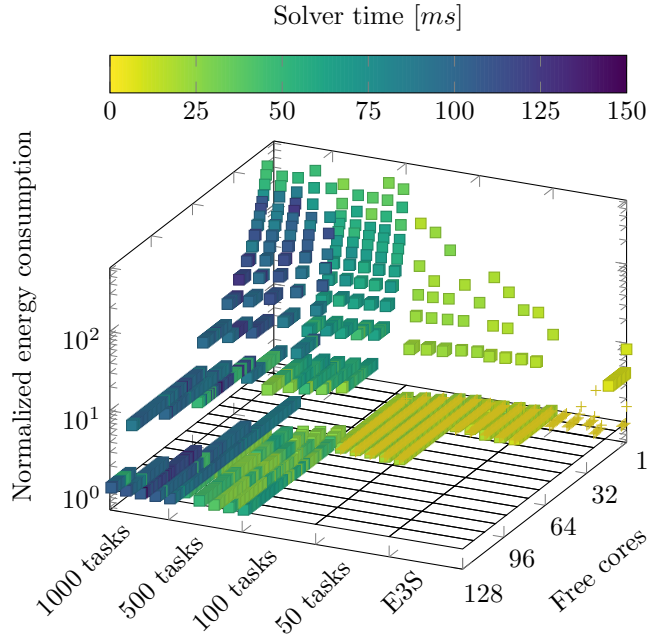
        For each of the 25 graphs, the number of cores has been varied between 1 and 128.
        In each design point, the energy of the found solution has been normalized to the optimal energy for the discrete speed case with no core constraints as determined by the \ilpmap approach.
       
        It can be seen that \roundsched is able to solve many instances of graphs with 50 to 100 tasks in less than 25ms.
        However, it does not find a solution for 4 out of 5 E3S graphs because of the tightness of deadlines assumed in these benchmarks and the assumptions made in~\Cref{roundtheo}.
        Finally, we omit to present and compare the solver times of \ilpsched as these start in the range of minutes even for the smallest and easiest problem instances. 
        Hence, we conclude this approach to be of no use to be applied on an MPSoC at run-time.
        %\bert{TODO: finish this with the final plot}

	%--------------------------------------------------------------------
	\section{Conclusions}
	%--------------------------------------------------------------------	
	
	%\todo[inline]{Answer the questions of the introduction; discuss possible future work / extensions}
	We have shown that for many task graphs of real-world applications, the graph structure 
	allows to determine energy-optimal speed assignments in the range of a $ms$   
	given real-time constraints by applying an algorithm called \spalgo in case tasks have been mapped 
    already to cores. For the more complex problem of additionally determining the task-to-core assignment
    and schedule of tasks on these cores, \new{even problem instances with few tasks cannot 
    be practically solved optimally at runtime}. Yet here, approximation algorithms have been analyzed and shown to offer
    affordable solving times to determine at least solutions with provable guarantees on the solution quality.
    % (Bertrand) I changed the phrasing which implied that we can approximate instances with tight deadlines, which we cannot

    In the future, we would like to extend our analysis of the ties between theory and 
    practice from homogeneous MPSoCs to systems with more diverse and complex 
    communication architectures. Moreover, the presented set of algorithms shall be integrated
    into a framework for run-time resource management on many-core systems
    that are required to stay within given bounds on execution time, energy and also other user-specific 
    requirement corridors.

	%--------------------------------------------------------------------	
	\bibliography{bib}

\begin{thebibliography}{10}

\bibitem{Pegasus}
Epigenomics dataset from the pegasus library.
\newblock
  \url{https://confluence.pegasus.isi.edu/display/pegasus/WorkflowGenerator}.
\newblock [Online; accessed 02-September-2019].

\bibitem{albers10}
Susanne Albers.
\newblock Energy-efficient algorithms.
\newblock {\em Communications of the ACM}, 53(5):86--96, 2010.

\bibitem{aupy}
Guillaume Aupy, Anne Benoit, Fanny Dufoss{\'e}, and Yves Robert.
\newblock Reclaiming the energy of a schedule: models and algorithms.
\newblock {\em Concurrency and Computation: Practice and Experience},
  25(11):1505--1523, 2013.

\bibitem{bampis}
Evripidis Bampis, Dimitrios Letsios, and Giorgio Lucarelli.
\newblock A note on multiprocessor speed scaling with precedence constraints.
\newblock In {\em Proceedings of the 26th ACM Symposium on Parallelism in
  Algorithms and Architectures (SPAA)}, pages 138--142, 2014.

\bibitem{Chen0K14}
Gang Chen, Kai Huang, and Alois Knoll.
\newblock Energy optimization for real-time multiprocessor system-on-chip with
  optimal {DVFS} and {DPM} combination.
\newblock {\em {ACM} Trans. Embedded Comput. Syst.}, 13:111:1--111:21, 2014.

\bibitem{ChoudhuryCK07}
Pravanjan Choudhury, P.~P. Chakrabarti, and Rajeev Kumar.
\newblock Online dynamic voltage scaling using task graph mapping analysis for
  multiprocessors.
\newblock In {\em 20th International Conference on {VLSI} Design}, pages
  89--94, 2007.

\bibitem{LangenJ09}
Pepijn~J. de~Langen and Ben H.~H. Juurlink.
\newblock Leakage-aware multiprocessor scheduling.
\newblock {\em Signal Processing Systems}, 57(1):73--88, 2009.

\bibitem{e3s}
R.~Dick.
\newblock Embedded system synthesis benchmarks suite {(E3S)}.
\newblock \url{http://ziyang.eecs.umich.edu/~dickrp/e3s/}.
\newblock [Online; accessed 02-September-2019].

\bibitem{GareyJohnson78}
M.R. Garey and D.S. Johnson.
\newblock {S}trong {N}{P}-completeness results: motivation, examples, and
  implications.
\newblock {\em J. Assoc. Comput. Mach.}, 25(3):499--508, 1978.

\bibitem{graham}
R.~L. Graham.
\newblock Bounds for certain multiprocessing anomalies.
\newblock {\em The Bell System Technical Journal}, 45(9):1563--1581, Nov 1966.

\bibitem{europar}
Abdou Guermouche, Loris Marchal, Bertrand Simon, and Fr{\'e}d{\'e}ric Vivien.
\newblock Scheduling trees of malleable tasks for sparse linear algebra.
\newblock In {\em European Conference on Parallel Processing}, pages 479--490.
  Springer, 2015.

\bibitem{iraniP05}
Sandy Irani and Kirk Pruhs.
\newblock Algorithmic problems in power management.
\newblock {\em SIGACT News}, 36(2):63--76, 2005.

\bibitem{KwonK05}
Woo{-}Cheol Kwon and Taewhan Kim.
\newblock Optimal voltage allocation techniques for dynamically variable
  voltage processors.
\newblock {\em {ACM} Trans. Embedded Comput. Syst.}, 4(1):211--230, 2005.

\bibitem{lawler1978sequencing}
Eugene~L Lawler.
\newblock Sequencing jobs to minimize total weighted completion time subject to
  precedence constraints.
\newblock In {\em Annals of Discrete Mathematics}, volume~2, pages 75--90.
  Elsevier, 1978.

\bibitem{Li08}
Keqin Li.
\newblock Performance analysis of power-aware task scheduling algorithms on
  multiprocessor computers with dynamic voltage and speed.
\newblock {\em {IEEE} Trans. Parallel Distrib. Syst.}, 19(11):1484--1497, 2008.

\bibitem{LiY05}
Minming Li and F.~Frances Yao.
\newblock An efficient algorithm for computing optimal discrete voltage
  schedules.
\newblock {\em {SIAM} J. Comput.}, 35(3):658--671, 2005.

\bibitem{MegowV18}
Nicole Megow and Jos{\'{e}} Verschae.
\newblock Dual techniques for scheduling on a machine with varying speed.
\newblock {\em {SIAM} J. Discrete Math.}, 32(3):1541--1571, 2018.

\bibitem{NelsonMMSNG11}
Andrew Nelson, Orlando Moreira, Anca~Mariana Molnos, Sander Stuijk, Ba~Thang
  Nguyen, and Kees Goossens.
\newblock Power minimisation for real-time dataflow applications.
\newblock In {\em 14th Euromicro Conference on Digital System Design,
  Architectures, Methods and Tools ({DSD} 2011)}, pages 117--124, 2011.

\bibitem{nesterovN94}
Yurii Nesterov and Arkadii Nemirovskii.
\newblock {\em Interior Point Polynomial Algorithms in Convex Programming}.
\newblock Society for Industrial and Applied Mathematics, Philadelphia, PA,
  1994.

\bibitem{prasmus}
G.~N.~Srinivasa Prasanna and Bruce~R. Musicus.
\newblock Generalized multiprocessor scheduling and applications to matrix
  computations.
\newblock {\em IEEE TPDS}, 7(6):650--664, 1996.

\bibitem{pruhs2008speed}
Kirk Pruhs, Rob van Stee, and Patchrawat Uthaisombut.
\newblock Speed scaling of tasks with precedence constraints.
\newblock {\em Theory of Computing Systems}, 43(1):67--80, 2008.

\bibitem{ShinK03}
Dongkun Shin and Jihong Kim.
\newblock Power-aware scheduling of conditional task graphs in real-time
  multiprocessor systems.
\newblock In {\em Proceedings of the 2003 International Symposium on Low Power
  Electronics and Design, 2003, Seoul, Korea, August 25-27, 2003}, pages
  408--413, 2003.

\bibitem{SinghDK13}
Amit~Kumar Singh, Anup Das, and Akash Kumar.
\newblock Energy optimization by exploiting execution slacks in streaming
  applications on multiprocessor systems.
\newblock In {\em The 50th Annual Design Automation Conference 2013 ({DAC}
  2013)}, pages 115:1--115:7, 2013.

\bibitem{Svensson11}
Ola Svensson.
\newblock Hardness of precedence constrained scheduling on identical machines.
\newblock {\em {SIAM} J. Comput.}, 40(5):1258--1274, 2011.

\bibitem{TariqW17}
Umair~Ullah Tariq and Hui Wu.
\newblock Energy-aware scheduling of periodic conditional task graphs on
  mpsocs.
\newblock In {\em Proceedings of the 18th International Conference on
  Distributed Computing and Networking}, page~13, 2017.

\bibitem{Ullman75}
J.D. Ullman.
\newblock {N}{P}-complete scheduling problems.
\newblock {\em J. Comput. System Sci.}, 10:384--393, 1975.

\bibitem{yaoDS95}
F.~Frances Yao, Alan~J. Demers, and Scott Shenker.
\newblock A scheduling model for reduced {CPU} energy.
\newblock In {\em Proc. of the 36th Annual Symposium on Foundations of Computer
  Science (FOCS 1995)}, pages 374--382, 1995.

\end{thebibliography}
	%--------------------------------------------------------------------
\end{document}